\documentclass[10pt]{amsart}
\usepackage[dvips]{graphicx}
\usepackage{bbm}
\usepackage{epsfig}
\usepackage{caption}
\usepackage{color}
\usepackage{amsthm}
\usepackage{amsmath}
\usepackage{amsfonts}
\usepackage{amssymb}
\usepackage{graphics,graphicx}

\usepackage{color}
\def\be{\begin{equation}}
\def\ee{\end{equation}}
\def\beq{\begin{eqnarray}}
\def\eeq{\end{eqnarray}}
\def\beqs{\begin{eqnarray*}}
\def\eeqs{\end{eqnarray*}}
\def\ea{\end{array}}
\def\ea{\end{array}}
\def\ds{\displaystyle}

\def\RR{\mathbb{R}}
\def\11{{\rm 1~\hspace{-1.5ex}1} }
\def\NN{\mathbb{N}}

\def\CC{\mathbb{C}}

\newcommand{\rfb}[1]{\mbox{\rm
   (\ref{#1})}\ifx\undefined\stillediting\else:\fbox{$#1$}\fi}

\makeatletter
\def\section{\@startsection {section}{1}{\z@}{-3.5ex plus -1ex minus
    -.2ex}{2.3ex plus .2ex}{\large\bf}}
\makeatother

\font\eufm=eufm10\font\eufms=eufm10\font\eufmss=eufm10\newfam\eufam
\textfont\eufam=\eufm\scriptfont\eufam=\eufms\scriptscriptfont\eufam=\eufmss

\usepackage{amsthm}
\swapnumbers
\newtheorem{theorem}{Theorem}[section]
\newtheorem{lemma}[theorem]{Lemma}
\newtheorem{corollary}[theorem]{Corollary}
\newtheorem{remark}[theorem]{Remark}

\newtheorem{proposition}[theorem]{Proposition}

\begin{document}
\thispagestyle{empty}
\title[Dispersive effects for the Schr\"{o}dinger equation in a tadpole graph]{Dispersive effects for the Schr\"{o}dinger equation on a tadpole graph}
\author{Felix Ali Mehmeti}
\address{Universit\'e de Valenciennes et du Hainaut Cambr\'esis,
LAMAV, FR CNRS 2956, Le Mont Houy, 59313 Valenciennes Cedex 9, France}
\email{felix.ali-mehmeti@univ-valenciennes.fr}
\author{Ka\"{\i}s Ammari}
\address{UR Analysis and Control of PDEs, UR13E564, Department of Mathematics, Faculty of Sciences of
Monastir, University of Monastir, 5019 Monastir, Tunisia}
\email{kais.ammari@fsm.rnu.tn}
\author{Serge Nicaise}
\address{Universit\'e de Valenciennes et du Hainaut Cambr\'esis,
LAMAV, FR CNRS 2956, Le Mont Houy, 59313 Valenciennes Cedex 9,
France}
\email{snicaise@univ-valenciennes.fr}

\date{}

\begin{abstract}   We consider the free Schr\"odinger group $e^{-it \frac{d^2}{dx^2}}$
on a
tadpole graph ${\mathcal R}$.
We first show that the time decay estimates $L^1
({\mathcal R}) \rightarrow L^\infty ({\mathcal R})$ is in $|t|^{-\frac12}$ with a constant independent of the length of the circle. Our proof is based on an appropriate decomposition of the kernel of the
resolvent.
 Further we derive a dispersive perturbation estimate, which proves
 that the solution on the queue of the tadpole converges uniformly,
 after compensation of the underlying time decay, to the solution
 of the Neumann half-line problem, as the circle shrinks to a point.
 To obtain this result, we suppose that the initial condition fulfills a high frequency cutoff.
\end{abstract}


\maketitle

\tableofcontents

\section{Introduction} \label{formulare}
A characteristic feature of the Schr\"odinger equation is the loss
of the localization of wave packets during evolution, the
dispersion. This effect can be measured by $L^{\infty}$-time decay,
which implies a spreading out of the solutions, due to the time
invariance of the $L^{2}$-norm. The well known fact that the free
Schr\"odinger group in $\RR^n$ considered as an operator family from
$L^{1}$ to $L^{\infty}$ decays exactly as $c \cdot t^{-n/2}$ follows
easily from the explicit knowledge of the kernel of this group
\cite[p. 60]{ReedSimonII}.

In this paper we derive analogous $L^{\infty}$-time decay estimates
for Schr\"odinger equations on the tadpole graph
(sometimes also called lasso graph).

Before a precise statement of our main result,
let us introduce some notation which will be used throughout the
rest of the paper.

Let $R_i,i=1,2,$ be two disjoint sets
identified with a closed path of measure equal to $L > 0$ for $R_2$ and to $(0,+\infty),$
for $R_1$, see figure \ref{fig1}. We set ${\mathcal R} := \ds \cup_{k=1}^2
\overline{R}_k$. We denote by $f = (f_k)_{k=1,2} =
(f_1,f_2)$ the functions on ${\mathcal R}$ taking their values in
$\CC$ and let $f_k$ be the restriction of $f$ to $R_k$.


\begin{center} \label{fig1}
\includegraphics[scale=0.80]{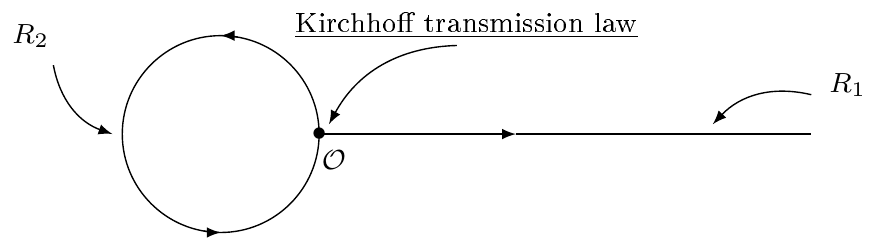}
\captionof{figure}{Tadpole graph}
\end{center}

Define the Hilbert space ${\mathcal H} = \ds \prod_{k=1}^2 L^2(R_k) = L^2 (\mathcal{R})$
with inner product
$$((u_k),(v_k))_{\mathcal H} = \ds \sum_{k=1}^2
(u_k,v_k)_{L^2(R_k)}$$
and introduce the following transmission
conditions (see \cite{kost,amamnic}):
\be
\label{t0}
(u_k)_{k=1,2} \in \ds \prod_{k=1}^2
C(\overline{R_k}) \; \hbox{satisfies} \; u_1(0) = u_2(0)=u_2(L),
\ee

\be
\label{t1}
(u_k)_{k=1,2} \in \ds \prod_{k=1}^2
C^1(\overline{R_k}) \; \hbox{satisfies} \; \ds \sum_{k=1}^2
\frac{du_k}{dx}(0^+) - \frac{du_2}{dx}(L^-)= 0.
\ee

Let $H : {\mathcal D}(H) \subset {\mathcal H} \rightarrow {\mathcal H}$ be the linear operator on
${\mathcal H}$ defined by :
$$
{\mathcal D}(H) = \left\{(u_k) \in \prod_{k=1}^2 H^2(R_k); \, (u_k)_{k=1,2} \;
\hbox{satisfies} \; \rfb{t0},\rfb{t1} \right\},
$$
$$
H (u_k) = (H_{k} u_k)_{k=1,2} = \left(- \frac{d^2 u_k}{dx^2} \right)_{k=1,2} = - \Delta_{{\mathcal R}}(u_k).
$$
This operator $H$ is self-adjoint and
its spectrum $\sigma(H)$ is equal to $[0,+\infty)$.
The self-adjointness and non-negativity of $H$ can be shown by Friedrichs extension (see \cite{felix}  for example), the fact that the spectrum is equal to the positive half-axis follows from Theorem \ref{res.id} below.

Here, we prove that the free Schr\"odinger group on the tadpole graph ${\mathcal R}$ satisfies the standard $L^1-L^\infty$ dispersive
estimate. More precisely, we will prove the following theorem.
\begin{theorem} \label{mainresult}
For all $t
\neq 0$, \be \label{dispest} \left\|e^{itH} P_{ac}
\right\|_{L^1({\mathcal R}) \rightarrow L^\infty({\mathcal R})} \leq C \,
\left|t\right|^{- 1/2}, \ee where $C$ is a positive constant independent of $L$ and $t$, $P_{ac}f$ is the projection onto the absolutely continuous spectral
subspace and $L^1 (\mathcal{R}) = \ds \prod_{k=1}^2 L^1(R_k), \, L^\infty (\mathcal{R}) = \ds \prod_{k=1}^2 L^\infty(R_k).$
\end{theorem}

 This means that the time decay is the same as the case of a line  \cite{ReedSimonII,GS},  a half-line  
\cite{wederb} or    star
shaped networks  \cite{admi1,ignat,amamnic}.  Note that the proof of this result is based on  an appropriate decomposition of the kernel of the
resolvent that in particular gives
 a full characterization of the  spectrum, made only of the point spectrum and of the absolutely continuous one;
 showing the absence of a singular continuous part.
 An important point is that this estimate is independent of the length $L$
of the circle, which also follows from the fact that the problem is scale invariant, as it
is shown in Remark \ref{scale invariance}.

Let $H_{0}$ be the negative laplacian on the half
line with Neumann boundary conditions. Then holds the following
dispersive perturbation estimate:
\begin{theorem} \label{dispersive perturbation estimate}
Let $0 \leq a < b < \infty$. Let $u_0 \in {\mathcal H} \cap
L^1(R_1)$ such that
\begin{equation}\label{support in queue}
    \hbox{\rm supp }u_0 \subset R_1 \ .
\end{equation}
Then   for  all $t
\neq 0$, we have
\begin{eqnarray}
\nonumber \parallel e^{itH} \mathbb{I}_{(a,b)}(H)u_0
&-&  e^{itH_0} \mathbb{I}_{(a,b)}(H_0) u_0  \parallel_{L^\infty (R_1)}\\
\nonumber & \leq &
t^{-1/2} L  \ 2 \sqrt{2} \
\Bigl(4(2\sqrt{b}-\sqrt{a}) + L(b-a)\Bigr) \
\left\| u_0 \right\|_{L^1(R_1)}.
  \end{eqnarray}
\end{theorem}
This last result implies that the solution of the Schr\"{o}dinger
equation on the queue $R_1$ of the tadpole with an upper frequency
cutoff tends uniformly to the solution of the half-line Neumann
problem with the same upper frequency cutoff, if the initial
condition has its support in the queue, after compensation of the
underlying $t^{-1/2}$-decay. In physical terms the frequency cutoff
makes that the localization of the signals is limited and thus they
have increasing difficulties to enter into the head of the tadpole.

Without the high frequency cutoff, this result would not be
possible,   as the problem is scale invariant.

The paper is organized as follows. The kernel of the resolvent
needed for the proof of Theorem \ref{mainresult}, is given in
section \ref{sec2}. Further all eigenfunctions of the free
Hamiltonian on the tadpole are constructed. They correspond to the
confined modes on the head of the tadpole, which do not interact
   with the queue. The interaction is described by the absolutely
continuous spectrum.
Technically the main point is a decomposition of the kernel of the
resolvent into a meromorphic term in the whole complex plane with
poles on the positive real axis and a term which is continuous on
the real line but discontinuous when crossing it.

The poles are shown to be the eigenvalues of the operator, the
continuous term creates the absolutely continuous spectrum. The
absence of further terms proves the absence of a singular continuous
spectrum. Note that in \cite{cacc} the eigenvalues and eigenvectors are given without this detailed analysis of the rest of the spectrum.

In section \ref{proofm} we give the proof of the main result of the
paper (Theorem \ref{mainresult}). To this end we replace the formula
of the resolvent into Stone's formula and prepare all terms for the
application of the Lemma of van der Corput.

In section \ref{limit L to 0} we prove the dispersive perturbation
estimate. This is based on the comparison of the kernel of the
resolvent of the tadpole operator on $R_1 \times R_2$ with the
kernel of the Schr\"{o}dinger operator with Neumann conditions on
the half-line.

\section{Kernel of the resolvent} \label{sec2}

Given $z\in \CC_+:=\{z\in \CC: \Im z>0\}$ and $g\in L^2({\mathcal R})$, we are looking
for $u \in {\mathcal D}(H)$ solution of
\[
- \Delta_{{\mathcal R}} u-z^2 u=g \hbox{ in } {\mathcal R}.
\]
Let us use the notation $\omega=-iz$.

Hence we look for $u$ in the form
$$
u_1(x) = \int_0^{+\infty} \frac{g_1(y)}{2\omega} \left( e^{-\omega |x-y|}-F_1(\omega) e^{-\omega (y+x)} \right) \,dy
$$
\be
\label{sergeu1}
- \int_0^L \frac{g_2(y)}{2\omega} \left(F_2(\omega) e^{-\omega (y+x)}+F_3(\omega) e^{\omega (y-x)} \right) \,dy,
\ee
$$
u_2(x) = \int_0^{+\infty} \frac{g_1(y)}{2\omega}\left(G_1(\omega) e^{-\omega (y+x)}+H_1(\omega) e^{-\omega (y-x)} \right) \,dy
$$
$$
 + \int_0^L \frac{g_2(y)}{2\omega} \left(e^{-\omega |x-y|}
+ G_2(\omega) e^{-\omega (y+x)}+G_3(\omega) e^{\omega (y-x)} \right.
$$
\be
\label{sergeu2}
\left.
+ \, H_2(\omega) e^{\omega (x-y)}+H_3(\omega) e^{\omega (x+y)} \right) \,dy,
\ee
where $F_i(\omega)$, $G_i(\omega)$ and  $H_i(\omega)$ , $i=1,2,3$ are constants fixed
in order to satisfy  \rfb{t0}, \rfb{t1}. Indeed from these expansion, we clearly see that
for $k=1$ or 2:
\[
-u_k''+\omega^2 u_k= g_k \hbox{ in } R_k.
\]
Now we see that the continuity condition \rfb{t0} is satisfied if and only if
\[
\begin{array}{lll}
G_1+H_1=G_1e^{-\omega L}+H_1e^{\omega L}=1-F_1,\\
1+G_2+H_2=G_2e^{-\omega L}+H_2e^{\omega L}=-F_2,\\
G_3+H_3=(1+G_3)e^{-\omega L}+H_3e^{\omega L}=-F_3,
\end{array}
\]
while Kirchhoff condition
\rfb{t1} holds
 if and only if
\[
\begin{array}{lll}
F_1+G_1(e^{-\omega L}-1)+H_1(1-e^{\omega L})=-1,\\
F_2+G_2(e^{-\omega L}-1)+H_2(1-e^{\omega L})=-1,\\
F_3+G_3(e^{-\omega L}-1)+H_3(1-e^{\omega L})=e^{-\omega L}.
\end{array}
\]

These equations correspond to three linear systems in $F_i, G_i, H_i$, $i=1,2,3$, whose associated matrix has a determinant $D(\omega)$ given by
\[
D(\omega)=e^{\omega L} (e^{-\omega L}-1) (e^{-\omega L}-3).
\]
Since this determinant is different from zero (as $\Im z>0$), we deduce the following expressions:

\beq\label{sergef1}
F_1(\omega)=1+\frac{2(e^{-\omega L}+1)}{e^{-\omega L}-3},\\
\label{sergeg1}
G_1(\omega)=-\frac{2}{e^{-\omega L}-3},\\
\label{sergeh1}
H_1(\omega)=-\frac{2e^{-\omega L}}{e^{-\omega L}-3},\\
\label{sergef2}
F_2(\omega)=\frac{2}{e^{-\omega L}-3},\\
\label{sergeg2}
G_2(\omega)=-\frac{e^{\omega L}}{D(\omega)},\\
\label{sergeh2}
H_2(\omega)=\frac{2-e^{-\omega L}}{D(\omega)},\\
\label{sergef3}
F_3(\omega)=-\frac{2e^{-2\omega L}}{e^{-\omega L}-3},\\
\label{sergeg3}
G_3(\omega)=\frac{e^{-\omega L}}{D(\omega)},\\
\label{sergeh3}
H_3(\omega)=\frac{e^{-\omega L}}{D(\omega)} (2e^{-\omega L}-3).
\eeq

Inserting these expressions in \rfb{sergeu1}-\rfb{sergeu2}, we have obtained the next result.

\begin{theorem} \label{theo1}
Let   $f \in \mathcal H$. Then, for $x \in {\mathcal R}$ and $z\in \mathbb{C}$ such that
$\Im z> 0$, we have
\be\label{resolvantformula} [R(z^2, H)f](x)= \int_{{\mathcal R}} K(x, x', z^2) f(x') \; dx',
\ee
where
the kernel $K$ is defined as follows:
\beq
\label{kernel11} K(x,y,z^2) &= &
    \frac{1}{2iz}  \left(e^{iz|x-y|} -F_1(-iz) e^{iz (x+y)}\right),   \forall x, y \in  {R_1},
\\
K(x,y,z^2) &=& -\frac{1}{2iz} \left(F_2(-iz) e^{iz (y+x)}+F_3(-iz) e^{-iz (y-x)}\right),  \forall x\in R_1, y\in R_2,
\label{kernel12}
\\
\label{kernel22}
K(x,y,z^2) &=&       \frac{1}{2iz}  \Big(e^{iz |x-y|}
+G_2(-iz) e^{iz (y+x)}+G_3(-iz) e^{-iz (y-x)}
\\
\nonumber
&+&H_2(-iz) e^{-iz (x-y)}+H_3(-iz) e^{-iz (x+y)} \Big),  \forall x, y \in  {R_2},
\\
K(x,y,z^2)&=&  \frac{1}{2iz} \left(G_1(-iz) e^{iz (y+x)}+H_1(-iz) e^{iz(y-x)} \right),  \forall x\in R_2, y\in R_1.
\label{kernel21}
\eeq
\end{theorem}

As usual, to obtain the resolution of the identity of $H$, we want to use the limiting absorption principle
that consists to pass to the limit in $K(x,y,z^2)$ as $\Im z$ goes to zero.
But in view of the presence of the factor $e^{iz L}-1$  in the denominator of $G_2, G_3, H_2, H_3$,
this limit is a priori not allowed. This factor comes from the circle $R_2$ and suggests that the point spectrum is distributed in the whole
continuous spectrum. This is indeed the case has the next results will show.

 \begin{lemma} \label{lemma1}
 For all $k\in \NN^*$, the number $\lambda_{2k}^2=\frac{4 k^2\pi^2}{L^2}$ is an eigenvalue of $H$
 of the associated eigenvector $\varphi^{(2k)}\in {\mathcal D}(H)$ given by
 \beq\label{vp1}
 \varphi^{(2k)}_1=0 \hbox{ in } R_1,\\
  \varphi^{(2k)}_2(x)=\frac{\sqrt{2}}{\sqrt{L}}\sin (\lambda_{2k} x), \forall x\in R_2.
  \label{vp2}
  \eeq
  Furthermore $H$ has no other eigenvalues.
\end{lemma}
\begin{proof}
The proof of the first assertion is direct since we readily check that $\varphi^{(2k)}$ defined by (\ref{vp1})-(\ref{vp2}) is indeed in ${\mathcal D}(H)$
and satisfies
$H\varphi^{(2k)}=\lambda_{2k}^2\varphi^{(2k)}$.

For the second assertion, we simply remark that if
$\varphi$ is an eigenvector of $H$ of eigenvalue $\lambda^2$, then
for $\lambda>0$, we have
\[
\varphi_1(x)= c_1 \sin (\lambda x)+c_2\cos (\lambda x), \forall x\in R_1,
\]
with $c_i\in \CC$. But the requirement that $\varphi_1$ belongs to $L^2(R_1)$ directly implies that
$c_1=c_2=0$. Hence $\varphi$ has to be in the form of the first assertion.
In the case $\lambda=0$,  $\varphi_1$ has to be zero and therefore
\[
\varphi_2(x)= c_1 +c_2x, \forall x\in R_2,
\]
with $c_i\in \CC$. By the continuity property at 0, we get
$c_1=c_1+c_2L=0$, hence $c_1=c_2=0$.
\end{proof}

\begin{remark}
We shall see below that the eigenvalues
$(\lambda_{2k}^2)_{k\in \mathbb{N}^*}$ are embedded in
the continuous spectrum with corresponding eigenfunctions
$\varphi^{(2k)}$  which are confined in the ring.
\end{remark}

At this stage we define the projection $P_{pp}$   on the closed subspace spanned by the $\varphi^{(2k)}$'s, namely
for any $f\in {\mathcal H}$, we set
\[
P_{pp} f=\sum_{k=0}^{+\infty} (f, \varphi^{(2k)})_{{\mathcal H}} \varphi^{(2k)}.
\]
Note that $P_{pp} f$ is different from $f$ on $R_2$
because $L^2(R_2)$ is spanned by the set of eigenvectors of the Laplace operator with Dirichlet boundary conditions at $0$ and $L$, that are the set $\{\varphi^{(\ell)}_2\}_{\ell\in \NN^*}$,
where
\[
\varphi^{(\ell)}_2(x)=\frac{\sqrt{2}}{\sqrt{L}}\sin (\lambda_{\ell} x), \forall x\in R_2,
\]
and $\lambda_{\ell} =\frac{\ell\pi}{L}$. Hence
\[
f- P_{pp} f=\sum_{k=0}^{+\infty} \left(\int_0^Lf(x)\varphi^{(2k+1)}_2(x)\,dx\right) \varphi^{(2k+1)}_2.
 \]

To show that our operator has no singular continuous 
spectrum, we shall split up the kernel $K(x,y, z^2)$ into a
meromorphic part $K_p$ in the whole complex plane with
poles at the points $\lambda_{2k}$ and a part $K_c$ which is
continuous on the real line but discontinuous when crossing it.

\begin{theorem}\label{lsplitkernel}
For all $z\in \CC$ such that $\Im z>0$, and all $x,y\in  {\mathcal
R}$, the kernel $K(x,y,z^2)$ defined in Theorem \ref{theo1} admits
the decomposition
\be\label{splitkernel} K(x,y,z^2)=K_c(x,y,z^2)+K_p(x,y,z^2), \ee
where for $x,y  \in R_2$ and $X = e^{izL}$ we have
\begin{eqnarray*}
  K_c(x,y,z^2) &=& - \frac{1}{2iz}
  \left( \frac{X+1}{X-3}   e^{iz(x-y)} + \frac{2X(X-1)}{X-3}
  e^{-iz(x+y)}- 2i\frac{X-1}{X-3}   \sin(zy)e^{-iz x}\right) \\
   && - \frac{1}{2iz}\left(1+\frac{2}{X-3} \right) \sin (zx)\sin%
   (zy).
\end{eqnarray*}
and
\begin{equation*}
K_p(x,y,z^2) =
\frac{\cos\left(\frac{zL}{2} \right)}{2z\sin\left(\frac{zL}{2}
\right)}  \sin (zx)\sin (zy)
\end{equation*}
The function $z  \mapsto K_c(x,y,z^2)$ is continuous on $\Im
z\geq 0$ except at $z=0$, while $z \mapsto K_p(x,y,z^2)$ is
meromorphic in $\CC$ with poles at the points $\lambda_{2k}$, $k\in
\NN^*$ and at $z=0$.

For $x \not\in R_2$ or $y \not\in R_2$
we have $K_p(x,y,z^2) = 0$ and $K_c(x,y,z^2)=K(x,y,z^2)$ as defined
in Theorem \ref{theo1}.
\end{theorem}

\begin{proof}
 The problem  only appears  for $x$ and $y$ in $R_2$, since in the other cases,
$K$ has no poles and therefore in that cases we simply take $K_p=0$.
Hence we need to perform this splitting for $x,y\in R_2$.

First we transform $G_2(-iz) e^{iz y}+G_3(-iz) e^{-izy}$
appropriately in order to bring out its meromorphic part that comes
from the zeroes of the factor $e^{izL}-1$ that  precisely correspond
to the points $z=\lambda_{2k}$. If we write for shortness
$X=e^{izL}$, we see that
  \beqs
  G_2(-iz) e^{iz y}+G_3(-iz) e^{-izy}&=&
  \frac{1}{(X-1)(X-3)}(X^2  e^{-izy}- e^{izy})
  \\
  &=&\frac{1}{(X-1)(X-3)}\left( (X^2-1)  e^{-izy}+  e^{-izy}- e^{izy}\right)
  \\
  &=&
  \frac{X+1}{X-3}   e^{-izy}-\frac{2i}{(X-1)(X-3)} \sin (zy).
  \eeqs

  In the same manner, we can show that
   \beqs
  H_2(-iz) e^{iz y}+H_3(-iz) e^{-izy}
  &=&
  \frac{2X(X-1)}{X-3}   e^{-izy}- 2i\frac{X-1}{X-3}   \sin(zy)+
  \frac{2i}{(X-1)(X-3)} \sin (zy).
  \eeqs

 These two expressions lead to
\beq
\label{splitting1}
 &&G_2(-iz) e^{iz (y+x)}+G_3(-iz) e^{-iz (y-x)}
+H_2(-iz) e^{-iz (x-y)}+H_3(-iz) e^{-iz (x+y)}
\\
\nonumber
&&= (G_2(-iz) e^{iz y}+G_3(-iz) e^{-izy}) e^{iz x}
+(H_2(-iz) e^{iz y}+H_3(-iz) e^{-izy}) e^{-iz x}
\\
&&= \frac{4}{(X-1)(X-3)}  \sin (zx)\sin (zy)+ K_1(x,y,z), \nonumber
\eeq
where
\be\label{defR1} K_1(x,y,z^2)=\frac{X+1}{X-3} e^{iz(x-y)} +
\frac{2X(X-1)}{X-3}   e^{-iz(x+y)}- 2i\frac{X-1}{X-3} \sin(zy)e^{-iz
x}, \ee
is clearly continuous up to $\Im z=0$. We are therefore reduced to
transform the factor $\frac{1}{(X-1)(X-3)}$. First  as its poles
correspond  to the case $X=1$, we can replace the factor $X-3$ by
$-2$,
indeed we have by partial fraction decomposition
\beqs \frac{1}{(X-1)(X-3)} &=&\frac{1}{2(X-3)}-\frac{1}{2(X-1)}.
\eeqs 
Recalling that $X=e^{izL}$, we have
\beqs
\frac{1}{X-1}=\frac{1}{e^{izL}-1}&=&\frac{e^{-\frac{izL}{2}}}{e^{\frac{izL}{2}}-{e^{-\frac{izL}{2}}}}
\\
&=&-\frac12+\frac{\cos\left(\frac{zL}{2} \right)}{2i\sin\left(\frac{zL}{2} \right)}.
\eeqs
Using these two identities into \eqref{splitting1}, we get
\beq
 \label{splitting2}
 &&G_2(-iz) e^{iz (y+x)}+G_3(-iz) e^{-iz (y-x)}
+H_2(-iz) e^{-iz (x-y)}+H_3(-iz) e^{-iz (x+y)}
\\&&=
i\frac{\cos\left(\frac{zL}{2} \right)}{\sin\left(\frac{zL}{2}
\right)}  \sin (zx)\sin (zy)+ K_2(x,y,z), \nonumber \eeq
where
\be\label{defR2} K_2(x,y,z^2)=K_1(x,y,z^2) + \left(1+\frac{2}{X-3}
\right) \sin (zx)\sin (zy), \ee
is clearly continuous up to $\Im z=0$.

Plugging this splitting into \eqref{kernel22}, we find
that
\[
K(x,y,z^2) =     \frac{\cos\left(\frac{zL}{2} \right)}{2z\sin\left(\frac{zL}{2} \right)}  \sin (zx)\sin (zy)+K_c(x,y, z^2), \forall x, y \in  {R_2},
\]
where
$
K_c(x,y, z^2)=\frac{1}{2iz}K_2(x,y,z),
$
is clearly continuous up to $\Im z=0$ except at $z=0$. This proves
\eqref{splitkernel} with
\be\label{defkernelpp}
K_p(x,y,z^2)= \frac{\cos\left(\frac{zL}{2} \right)}{2z\sin\left(\frac{zL}{2} \right)}  \sin (zx)\sin (zy),
\ee
which is holomorphic on $\CC$ with poles at $z=0$ and $z=\lambda_{2k}$ as stated.
\end{proof}

With these notations, we are able to give the expression of the resolution of the identity $E$ of $H$.

\begin{theorem} \label{res.id}
Take $f,g \in {\mathcal H} $ with a compact support and let $0< a < b < + \infty$. Then for any holomorphic  function $h$   on the complex plane, we have
\beq\label{resolid}
\nonumber
( h(H) E(a,b)f,g)_{\mathcal H}&=&  -
\frac{1}{\pi}
  \int\limits_{\mathcal R}    \Big(\int_{(a,b) } h(\lambda)
        \int\limits_{\mathcal R} f(x')
          \Im K_c(x,x',\lambda) \; dx'
           \; d\lambda\Big) \bar g(x)dx
           \\
           &+&\sum_{k\in \NN^*: a<\lambda_{2k}^2<b}  h(\lambda_{2k}^2) (f, \varphi^{(2k)})_{{\mathcal H}} (\varphi^{(2k)}, g)_{{\mathcal H}},
\eeq
where for all $\lambda>0$, $K_c(x,x',\lambda)$ is defined  in Lemma \ref{lsplitkernel}.
\end{theorem}
\begin{proof}
First by
Stone's formula, we have (see for instance     Lemma~3.13 of \cite{fam2}
or  Proposition 4.5 of \cite{FAMetall10})
\[
(h(H)E(a,b)f, g)_{{\mathcal H}}=\lim_{\delta\rightarrow 0}\lim_{\varepsilon\rightarrow 0}
\frac{1}{2i\pi}
\left( \int_{a+\delta}^{b-\delta} \left[ h(\lambda) R(\lambda-i\varepsilon, H)-R(\lambda+i\varepsilon, H)\right]  d\lambda   f,g\right)_{{\mathcal H}}.
\]
First using Lemma \ref{lsplitkernel}, we can write
\[
(h(H)E(a,b)f, g)_{{\mathcal H}}=I_c+I_p,
\]
where
\beqs
I_c&=&\lim_{\delta\rightarrow 0}\lim_{\varepsilon\rightarrow 0}
\frac{1}{2i\pi}
\left( \int_{a+\delta}^{b-\delta} h(\lambda)\left[ R_c(\lambda-i\varepsilon, H)-R_c(\lambda+i\varepsilon, H)\right]  d\lambda   f,g\right)_{{\mathcal H}}\\
I_p&=&\lim_{\delta\rightarrow 0}\lim_{\varepsilon\rightarrow 0}
\frac{1}{2i\pi}
\left( \int_{a+\delta}^{b-\delta} h(\lambda)\left[  R_p(\lambda-i\varepsilon, H)-R_p(\lambda+i\varepsilon, H)\right]  d\lambda   f,g\right)_{{\mathcal H}}.
\eeqs
where $R_p$ (resp. $R_c$) is the operator corresponding to the kernel $K_p$ (resp. $K_c$).

As $R_c(\lambda-i\varepsilon, H)=\overline{R_c(\lambda+i\varepsilon, H)}$, we can write
\beqs
I_c&=&-\frac{1}{\pi}\lim_{\delta\rightarrow 0}\lim_{\varepsilon\rightarrow 0}
\left( \int_{a+\delta}^{b-\delta}   h(\lambda) \Im R_c(\lambda+i\varepsilon, H)   d\lambda     f,g\right)_{{\mathcal H}}
\\
&=&-\frac{1}{\pi}\lim_{\delta\rightarrow 0}\lim_{\varepsilon\rightarrow 0}
  \int_{a+\delta}^{b-\delta} h(\lambda)(\Im R_c(\lambda+i\varepsilon, H)      f,g)_{{\mathcal H}}
  d\lambda
  \\
  &=&-\frac{1}{\pi}\lim_{\delta\rightarrow 0}\lim_{\varepsilon\rightarrow 0}
  \int_{a+\delta}^{b-\delta}h(\lambda)\int_{\mathcal R} \int_{\mathcal R}\Im K_c(x,x',\lambda+i\varepsilon)      f(x') dx' g(x) dx
  d\lambda.
\eeqs

At this stage, we take advantage of Theorem \ref{theo1} and Lemma \ref{lsplitkernel}.
First by (\ref{kernel11}) to (\ref{kernel21}) and by \eqref{splitkernel}, we see that
\[
K_c(x,y, \lambda+i\varepsilon)\longrightarrow K_c(x,y, \lambda), \hbox{ as } \varepsilon\longrightarrow0.
\]
 Furthermore  we see
 that
 \[
| K_c(x,y, \lambda+i\varepsilon)|\leq \frac{C}{|\lambda+i\varepsilon|},
\forall \lambda\in (a,b), x,y \in {\mathcal R},
\]
for $x,y\in \mathcal R$ for some positive constant $C$ independent of $x,y$.
 This allows to pass to the limit in $\varepsilon\longrightarrow0$ by using the convergence dominated theorem
 to obtain that
 \beqs
\lim_{\delta\rightarrow 0}\lim_{\varepsilon\rightarrow 0}
  \int_{a+\delta}^{b-\delta}h(\lambda)\int_{\mathcal R} \int_{\mathcal R}\Im K_c(x,x',\lambda+i\varepsilon)       f(x') dx' g(x) dx
  d\lambda
  \\
  =\int_{a}^{b}h(\lambda)\int_{\mathcal R} \int_{\mathcal R}\Im K_c(x,x',\lambda)     f(x') dx' g(x) dx
  d\lambda.
  \eeqs

 Hence it remains to manage the term $I_p$.
As before we can write
 \[
 I_p=\frac{1}{2i\pi} \lim_{\delta\rightarrow 0}\lim_{\varepsilon\rightarrow 0}
  \int_{a+\delta}^{b-\delta}\int_{R_2} \int_{R_2} h(\lambda) (K_p(x,x',\lambda-i\varepsilon)-K_p(x,x',\lambda+i\varepsilon) )     f(x') dx' g(x) dx
 d\lambda.
\]

First by Lemma \ref{convdominee} below, we can show that

\beqs
&&\lim_{\varepsilon\rightarrow 0}
  \int_{a+\delta}^{b-\delta}\int_{R_2} \int_{R_2} h(\lambda) (K_p(x,x',\lambda-i\varepsilon)-K_p(x,x',\lambda+i\varepsilon) )     f(x') dx' g(x) dx d\lambda
  \\
 && =
  \lim_{\varepsilon\rightarrow 0}
  \int_{a+\delta}^{b-\delta}\int_{R_2} \int_{R_2}  ( h(\lambda-i\varepsilon)K_p(x,x',\lambda-i\varepsilon)-h(\lambda+i\varepsilon)K_p(x,x',\lambda+i\varepsilon) )     f(x') dx' g(x) dx  d\lambda.
  \eeqs

Now, for a fixed $\delta>0$, we can always assume that $a+\delta$ and $b-\delta$ are always different from $\lambda_{2k}$
and therefore
\[
\lim_{\varepsilon\rightarrow 0}
  \int_{-\varepsilon}^{\varepsilon}\int_{R_2} \int_{R_2}  h(c_\delta+iy) K_p(x,x',c_\delta+iy)    f(x') dx' g(x) dx
 dy=0,
 \]
 for $c_\delta=a+\delta$ or $b-\delta$.
Consequently if we define the contour
$C_{\varepsilon,\delta}$ by the lines $\lambda-i\varepsilon$, $\lambda+i\varepsilon$
with $\lambda\in (a+\delta,b-\delta)$ and $a+\delta+iy$, $b-\delta+iy$, with $y\in (-\varepsilon,\varepsilon)$, we deduce that
\[
I_p= -\lim_{\delta\rightarrow 0}\lim_{\varepsilon\rightarrow 0}
\frac{1}{2i\pi}
 \int_{C_{\varepsilon,\delta}} \int_{R_2} \int_{R_2}  h(\lambda) K_p(x,x',\lambda)  f(x') dx' g(x) dx
 d\lambda.
 \]
 Now reminding the  expression \eqref{defkernelpp}, we perform the change of variable
 $\mu= \sqrt{\lambda}$, that leads to
 \beqs
&& -\frac{1}{2i\pi}  \int_{C_{\varepsilon,\delta}} \int_{R_2}
\int_{R_2}  h(\lambda)K_p(x,x',\lambda)  f(x') dx' g(x) dx
 d\lambda
 \\
&& =
  \frac{1}{2i\pi} \int_{D_{\varepsilon,\delta}}  2 \mu \ h(\mu^2)\int_{R_2} \int_{R_2}  \frac{\cos\left(\frac{ \mu L}{2} \right)}{\sin\left(\frac{ \mu L}{2} \right)}  \sin (\mu x)\sin (\mu x') f(x') dx' g(x) dx
 d\mu,
 \eeqs
 where $D_{\varepsilon,\delta}=\{\sqrt{\lambda}: \lambda \in C_{\varepsilon,\delta}\}$ is another contour in the complex plane.
 At this stage we can apply  residue theorem and deduce, after simple calculations of the residues,  that
 \beqs
 &&-\frac{1}{2i\pi}
 \int_{C_{\varepsilon,\delta}} \int_{R_2} \int_{R_2}  h(\lambda) K_p(x,x',\lambda)  f(x') dx' g(x) dx
 d\lambda
 \\
&&=\frac{2}{L}
 \sum_{k\in \NN^*: a+\delta<\lambda_{2k}^2<b-\delta} h(\lambda_{2k}^2) \left(\int_{R_2}  f(x') \sin (\lambda_{2k} x') dx'\right)
 \left(\int_{R_2} g(x) \sin (\lambda_{2k} x) dx\right)
 \\
 &&=
 \sum_{k\in \NN^*: a+\delta<\lambda_{2k}^2<b-\delta} h(\lambda_{2k}^2) (f, \varphi^{(2k)})_{{\mathcal H}} (\varphi^{(2k)}, g)_{{\mathcal H}}.
 \eeqs

 This proves the result by passing to the limit in $\delta\to 0$.
\end{proof}

\begin{corollary}The operator $H$ has no singular spectrum and
\[
P_{ac} f=f-P_{pp}f, \forall f\in {\mathcal H}.
\]
\end{corollary}

\begin{lemma}
\label{convdominee}
Under the assumption of the previous theorem, one has
\beqs
&&\lim_{\varepsilon\rightarrow 0}
  \int_{a+\delta}^{b-\delta}\int_{R_2} \int_{R_2} h(\lambda) K_p(x,x',\lambda\pm i\varepsilon) f(x') dx' g(x) dx d\lambda
\\
&&  =
  \lim_{\varepsilon\rightarrow 0}
  \int_{a+\delta}^{b-\delta}\int_{R_2} \int_{R_2}  h(\lambda\pm i\varepsilon)K_p(x,x',\lambda\pm i\varepsilon)     f(x') dx' g(x) dx d\lambda.
  \eeqs
\end{lemma}
\begin{proof}
The proof is based on the use of the Lebesgue's convergence theorem.
Let us prove it   in the case $\lambda+ i\varepsilon$.
Writing
\beqs
\int_{a+\delta}^{b-\delta}  h(\lambda) K_p(x,x',\lambda+ i\varepsilon) d\lambda
&=&
\int_{a+\delta}^{b-\delta}  (h(\lambda)-h(\lambda+ i\varepsilon)) K_p(x,x',\lambda+i\varepsilon) d\lambda
\\
&+&
\int_{a+\delta}^{b-\delta}  h(\lambda+ i\varepsilon) K_p(x,x',\lambda+i\varepsilon) d\lambda,
\eeqs
we only need to show that
\[
\lim_{\varepsilon\rightarrow 0}
  \int_{a+\delta}^{b-\delta}\int_{R_2} \int_{R_2}   (h(\lambda)-h(\lambda+ i\varepsilon))K_p(x,x',\lambda+i\varepsilon)     f(x') dx' g(x) dx d\lambda=0.
  \]

 Since for any $\lambda\ne \lambda_{2k}^2$, one has
 \[
 (h(\lambda)-h(\lambda+ i\varepsilon))K_p(x,x',\lambda+i\varepsilon)\to 0, \hbox{ as } \varepsilon\to 0,
 \]
 to apply Lebesgue's convergence theorem, it suffices, for instance, to show that
 $(h(\lambda)-h(\lambda+ i\varepsilon))K_p(x,x',\lambda+i\varepsilon)$ is uniformly bounded (in $\varepsilon$).
But  in view of the definition  \eqref{defkernelpp} of $K_p$, we only need to estimate
the ratio
\[
q(\lambda,\varepsilon):=\frac{h(\lambda)-h(\lambda+ i\varepsilon)}{\sin\left(\frac{\sqrt{\lambda+ i\varepsilon} L}{2} \right)}.
\]

Now using a Taylor expansion, we can say that for $z$ small enough, say $|z|<\eta$, we have
\[
\sin\left(\frac{\sqrt{\lambda_{2k}^2+z} L}{2} \right)=\cos(k\pi) \frac{zL}{2\lambda_{2k}}+o(z).
\]
Hence applying this property to $\lambda-\lambda_{2k}^2+ i\varepsilon$ for one $k\in \NN^*$, we find that
\[
\sin\left(\frac{\sqrt{\lambda+ i\varepsilon} L}{2} \right)\sim  \cos(k\pi) \frac{L}{2\lambda_{2k}}  (\lambda-\lambda_{2k}^2+ i\varepsilon),
\]
for $(\lambda-\lambda_{2k}^2)^2+ \varepsilon^2<\eta^2$. Hence for $|\lambda-\lambda_{2k}^2|<\frac{\eta}{2}$
and $ \varepsilon<\frac{\eta}{2}$, we deduce that
\[
\left|\sin\left(\frac{\sqrt{\lambda+ i\varepsilon)} L}{2} \right) \right|\sim    \frac{L}{2\lambda_{2k}}
 |\lambda-\lambda_{2k}^2+ i\varepsilon|\geq \frac{L}{2\lambda_{2k}}
 \varepsilon.
\]
Since $h$ is holomorphic, we clearly have
\[
|h(\lambda)-h(\lambda+ i\varepsilon)| \leq C \varepsilon, \forall \lambda\in [a,b],
\]
for some $C>0$ (independent of $\lambda$ and $\varepsilon$). Therefore
for any $\lambda$ such that
$|\lambda-\lambda_{2k}^2|<\frac{\eta}{2}$ for some $k\in \NN^*$,
and for any $ \varepsilon<\frac{\eta}{2}$, one has
\[
|q(\lambda,\varepsilon)|\leq C',
\]
for some $C'>0$ (independent of $\lambda$ and $\varepsilon$).

It then remains to treat the case where $|\lambda-\lambda_{2k}^2|\geq\frac{\eta}{2}$ for all $k\in \NN^*$.
But in that case, as
\[
\sqrt{\lambda+i\varepsilon}-\lambda_{2k}=\frac{\lambda+i\varepsilon-\lambda_{2k}^2}{\sqrt{\lambda+i\varepsilon}-\lambda_{2k}},
\]
 for $\varepsilon$ small enough, we get
 \[
|\sqrt{\lambda+i\varepsilon}-\lambda_{2k}|\geq C" \eta,
\]
for some $C">0$ (independent of $\lambda$ and $\varepsilon$).
This implies that
\[
\left|\sin\left(\frac{\sqrt{\lambda+ i\varepsilon)} L}{2} \right) \right| \geq \delta,
\]
for some $\delta>0$ (independent of $\lambda$ and $\varepsilon$)
and again implies that $q(\lambda,\varepsilon)$ is uniformly bounded in that case as well.

\end{proof}

\section{Proof of the main result} \label{proofm}

For any $0<a<b<\infty$,
by Theorem \ref{res.id}
 for any $x,y\in {\mathcal R}$, we have found the following
 expression for the kernel of the operator
 $e^{itH}\mathbb{I}_{(a,b)}P_{ac}$:
$$
 \int_0^{+ \infty} e^{it\lambda} \mathbb{I}_{(a,b)}(\lambda) E_{ac}(d\lambda)(x,y)=
-\frac{1}{\pi}\int_a^b e^{it\lambda} \mathbb{I}_{(a,b)}(\lambda) \Im K_c(x,y,\lambda) d\lambda,
$$
and by the change of variables $\lambda=\mu^2$, we get
$$
 \int_0^{+\infty} e^{it\lambda} \mathbb{I}_{(a,b)} E_{ac}(d\lambda)(x,y)=
-\frac{2}{\pi}\int_{\sqrt{a}}^{\sqrt{b}} e^{it\mu^2}
\mathbb{I}_{(a,b)}(\mu^2) \Im K_c(x,y,\mu^2) \mu \
 d\mu.
$$
Now recalling the definition of $K_c$, we have to distinguish between the following  cases:
\begin{enumerate}
\item
If $x, y\in R_1$, then
%
%
$$
2i\mu K_c(x,y,\mu^2)=e^{i\mu|x-y|} -F_1(-i\mu) e^{i\mu (x+y)}.
$$
 Hence in that case, we have to estimate
 \[
\left|
\int_{\sqrt{a}}^{\sqrt{b}} e^{it\mu^2}   e^{i\mu|x-y|}   d\mu\right|,
\]
and
\[
\left|
\int_{\sqrt{a}}^{\sqrt{b}} e^{it\mu^2}
 F_1(-i\mu) e^{i\mu (x+y)} d\mu\right|.
\]
The first term is directly estimated by van der Corput Lemma
\cite[p. 332]{Zygmung:99}:
 \[
\left| \int_{\sqrt{a}}^{\sqrt{b}} e^{it\mu^2}   e^{i\mu|x-y|}
d\mu\right|\leq   \frac{4 \sqrt{2}}{\sqrt{t}}, \,
\forall \, t
> 0.
\]

For the second term by using (\ref{sergef1}), we have
\beqs
\left|
\int_{\sqrt{a}}^{\sqrt{b}} e^{it\mu^2}
 F_1(-i\mu) e^{i\mu (x+y)} d\mu\right|
& \leq&
\left| \int_{\sqrt{a}}^{\sqrt{b}} e^{it\mu^2}
e^{i\mu (x+y)} d\mu\right|
\\
&+&2\left|\int_{\sqrt{a}}^{\sqrt{b}} e^{it\mu^2}
  e^{i\mu (x+y)} \frac{e^{i\mu L}+1}{e^{i\mu  L}-3} d\mu\right|.
\eeqs
Again the first term of this right-hand side is  estimated by van der Corput Lemma, while for the second one we use the Neumann series

\[
\frac{1}{e^{i\mu L}-3}=-\frac{1}{3(1-\frac{e^{i\mu  L}}{3})}
=-\frac{1}{3}\sum_{k=0}^{+\infty} \frac{e^{ik\mu  L}}{3^k},
\]
to obtain (owing to Fubini's theorem)
\[
\left|\int_{\sqrt{a}}^{\sqrt{b}} e^{it\mu^2}
  e^{i\mu (x+y)} \frac{e^{i\mu  L}+1}{e^{i\mu  L}-3} d\mu\right|
  \leq \frac{1}{3} \sum_{k=0}^{+\infty}
\frac{1}{3^k} \left|\int_{\sqrt{a}}^{\sqrt{b}} e^{it\mu^2}
  e^{i\mu (x+y)+ik\mu } (e^{i\mu  L}+1)  d\mu\right|.
  \]
  Again applying van der Corput Lemma at each term we obtain
\[
\left|\int_{\sqrt{a}}^{\sqrt{b}} e^{it\mu^2}
  e^{i\mu (x+y)} \frac{e^{i\mu  L}+1}{e^{i\mu  L}-3} d\mu\right|
  \leq
  \frac{8 \sqrt{2}}{3 \sqrt{t}}\sum_{k=0}^{+\infty}
\frac{1}{3^k}=\frac{4\sqrt{2}}{\sqrt{t}}.
\]

All together we have proved that \be\label{term11} \left|\int_0^{+
\infty} e^{it\lambda} \mathbb{I}_{(a,b)}(\lambda)
E_{ac}(d\lambda)(x,y)\right| \leq
\frac{9\sqrt{2}}{\sqrt{t}},
\forall x,y\in R_1, t>0. \ee

\item
If $x\in R_2$ and $y\in R_1$, then
$$
2i\mu K_c(x,y,\mu^2)= G_1(-i\mu) e^{i\mu (y+x)}+H_1(-i\mu)
e^{i\mu(y-x)}
$$
and in view of the form of $G_1$ and $H_1$, the same arguments as before imply that
\be\label{term21}
\left|\int_0^{+ \infty} e^{it\lambda} \mathbb{I}_{(a,b)}(\lambda) E_{ac}(d\lambda)(x,y)\right|
\leq \frac{C}{\sqrt{t}}, \forall x\in R_2,y\in R_1, t>0,
\ee
where $C>0$ is independent of $x,y,t, a, b$ and $L$.

\item
If $x\in R_1$ and $y\in R_2$, then
$$
2i\mu K_c(x,y,\mu^2)=-(F_2(-i\mu) e^{i\mu (y+x)}+F_3(-i\mu)
e^{-i\mu(y-x)})
$$
and from the form of $F_2$ and $F_3$, we deduce as before that

\be\label{term12}
\left|\int_0^{+ \infty} e^{it\lambda} \mathbb{I}_{(a,b)}(\lambda) E_{ac}(d\lambda)(x,y)\right|
\leq \frac{C}{\sqrt{t}}, \forall x\in R_1,y\in R_2, t>0,
\ee
where $C>0$ is independent of $x,y,t, a, b$ and $L$.

\item
If $x,y\in R_2$, then owing to Lemma \ref{lsplitkernel}, we have
$$
2i\mu K_c(x,y,\mu^2)=K_2(x,y,\mu),
$$
with $K_2$ defined by \eqref{defR2}. But again the form of $K_2$ (and of $K_1$)
allows to deduce as before that

\be\label{term22}
\left|\int_0^{+ \infty} e^{it\lambda} \mathbb{I}_{(a,b)}(\lambda) E_{ac}(d\lambda)(x,y)\right|
\leq \frac{C}{\sqrt{t}}, \forall x,y\in R_2, t>0,
\ee
where $C>0$ is independent of $x,y,t, a, b$ and $L$.
\end{enumerate}

The estimates (\ref{term11}) to (\ref{term22}) imply that for all $f\in L^2({\mathcal R})\cap L^1({\mathcal R})$

\be\label{estab}
\|e^{itH}  \mathbb{I}_{(a,b)}(H) P_{ac} f\|_\infty \leq \frac{2C}{\sqrt{t}} \| f\|_1, \forall  t>0,
\ee
where $C>0$ is independent of $t, a, b$ and $L$.

As $e^{itH}  \mathbb{I}_{(a,b)}(H) P_{ac} f$ converges to $e^{itH}   P_{ac} f$ in $ L^2({\mathcal R})$
as $a\to 0$ and $b\to\infty$, by extracting a subsequence, we have
that
\[
e^{itH}  \mathbb{I}_{(a,b)}(H) P_{ac} f \to e^{itH}   P_{ac} f \hbox{ a.e.},
\]
and therefore (\ref{estab}) implies that
\be\label{est0infty}
\|e^{itH}    P_{ac} f\|_\infty \leq \frac{2C}{\sqrt{t}} \| f\|_1, \forall  t>0,
\ee
where $C>0$ is independent of $t$ and $L$. By density we conclude the proof of Theorem \ref{mainresult}.

\begin{remark} \label{scale invariance}
{\rm The above proof underlines that the constant $C$ appearing in
the $L^1-L^\infty$ estimate \eqref{dispest} is independent of the
length $L$ of $R_2$. But this independence can be   proved with the
help of the following scaling argument. Let $u$ be a solution of the
Schr\"odinger equation  on the tadpole graph ${\mathcal R}$ with
initial datum $u_0$, i.e.,  solution of \beqs
\frac{du_k}{dt}+i\frac{d^2 u_k}{dx^2}=0 \hbox{ in } R_k\times \RR,
k=1,2,
\\
u_1(0, t)=u_2(0^+, t)=u_2(L^-, t), \hbox{ in }  \RR,
\\
 \sum_{k=1}^2
\frac{du_k}{dx}(0^+,t) - \frac{du_2}{dx}(L^-,t)= 0,  \hbox{ in }  \RR,
\\
u(\cdot, 0)=u_0,  \hbox{ in }{\mathcal R}.
\eeqs
Then we perform the change of variables $x= L\hat x$, $t=L^2 \hat t$
that transform ${\mathcal R}\times \RR$ into $\hat{\mathcal R}\times \RR$, where
$\hat{\mathcal R}=\hat R_2\cup \hat R_1$, $\hat R_1=(0,+\infty)$ and $\hat R_2$ is a closed path of length 1.
Hence by setting $\hat u_0(\hat x)=u_0(x)$
and $\hat u(\hat x, \hat t)=u(x, t)$, we see that $\widehat{P_{ac} u_0}=\hat P_{ac} \hat u_0$
and that $\hat u$ is solution of the Schr\"odinger equation  on the tadpole graph $\hat{\mathcal R}$ with initial datum $\hat u_0$.
Hence applying the estimate \eqref{dispest} on $\hat{\mathcal R}$, we find that
\[
\| e^{i\hat t \hat H} \hat P_{ac} \hat u_0\|_{L^\infty(\hat{\mathcal R})} \leq \frac{\hat C}{\sqrt{\hat t}} \| \hat u_0\|_{L^1(\hat{\mathcal R})}, \forall  \hat t>0,
\]
 where $\hat C$ is a positive constant independent of  $t$.
 As $\| e^{it  H}   P_{ac}   u_0\|_{L^\infty({\mathcal R})}=
 \| e^{i\hat t \hat H} \hat P_{ac} \hat u_0\|_{L^\infty(\hat{\mathcal R})}$
 and
$ \|u_0\|_{L^1({\mathcal R})}= L\| \hat u_0\|_{L^1(\hat{\mathcal R})},$ we find that, for all $t>0$
\beqs
\| e^{it  H}   P_{ac}   u_0\|_{L^\infty({\mathcal R})}&=&\| e^{i\hat t \hat H} \hat P_{ac} \hat u_0\|_{L^\infty(\hat{\mathcal R})}
\\
&\leq& \frac{\hat C}{\sqrt{\frac{t}{L^2}}} \| \hat u_0\|_{L^1(\hat{\mathcal R})}
\\
&\leq& \frac{\hat C}{\sqrt{t}} \| u_0\|_{L^1({\mathcal R})}.
\eeqs
This proves that \eqref{dispest} holds on ${\mathcal R}$ with $C\leq\hat C$.
Since the converse implication also holds, we have shown that $C=\hat C$ in \eqref{dispest}.
Therefore if \eqref{dispest} holds for a certain $C$ and $L_0$, then it holds for all $L$
with the same $C$.
}
\end{remark}


\section{The shrinking circle limit for initial conditions in frequency bands} \label{limit L to 0}

    In this section we consider the limit of the solution of the
Schr\"{o}dinger equation on the tadpole as the circumference of the
circle tends to zero. To obtain a result, we need the crucial
hypothesis, that the initial condition has an upper cutoff
frequency.
     We shall use the formulas for the kernel
$K(x,y,\lambda)$ of the resolvent $(H-\lambda)^{-1}$ of the negative
laplacian $H$ on the tadpole and the kernel $K_{0}(x,y,\lambda)$ of
the resolvent $(H_{0}-\lambda)^{-1}$ of the negative laplacian
$H_{0}$ on the half line with Neumann boundary conditions: let us
recall that by equation \eqref{kernel11} we have

$$ K(x,y,z^2) =
    \frac{1}{2iz}  \left(e^{iz|x-y|} -1+\frac{2(e^{izL}+1)}{e^{izL}-3} e^{iz (x+y)}\right),$$
for $\Im z >0$ and $x,y \in R_1 \cong (0,\infty)$. Further we have

$$ K_0(x,y,z^2) =
    \frac{1}{2iz}  \left(e^{iz|x-y|} + e^{iz (x+y)}\right),$$
for $\Im z >0$ and $x,y \geq 0$, which can be checked by direct
calculations. Inserting this expression in Stone's formula and
applying the limiting absorption principle yields

\begin{eqnarray*}
 \bigl( e^{itH_0} \mathbb{I}_{(a,b)}(H_0) u_0\bigr)(x)&=&
 \frac{2}{\pi} \int_{\sqrt{a}}^{\sqrt{b}} e^{it \mu^2} \cos(\mu x) \Bigl( \int_{0}^{\infty} \cos(\mu y) u_0(y) \  dy \Bigr) d\mu\\
  &=&  \mathcal{C}^{-1} \Bigl[ \mathbb{I}_{(\sqrt{a},\sqrt{b})}(\mu)e^{it \mu^2} (\mathcal{C}u_0)(\mu)\Bigr](x) =:
  u(t,x)  \label{solution formula half line}
\end{eqnarray*}
where  $\mathcal{C}$ and $\mathcal{C}^{-1}$ are the cosine transform
and its inverse. This last representation makes it easy to check,
that for $-\infty < a < b < \infty$ the function $u$ is smooth and
satisfies
\begin{equation}\label{Neumann}
\left\{
  \begin{array}{ll}
    i u_t - u_{xx} = 0, & \hbox{} t,x \geq 0\\
    u_x(t,0) = 0, & \hbox{}  t\geq 0 \\
    u(0,\cdot) = u_0. & \hbox{}
  \end{array}
\right.
\end{equation}
Now we calculate the difference of the kernels of the resolvents of
the tadpole problem on its queue and of the half line problem with
the Neumann boundary condition:

\begin{proposition}\label{formula difference}
For $x,y \in R_1 \cong (0,\infty)$ and $\mu > 0$ we have
\begin{equation}\label{difference}
K(x,y,\mu^2) - K_0(x,y,\mu^2)
=
- \frac{2(e^{izL}+1)}{e^{izL}-3} e^{iz (x+y)}.
\end{equation}
\end{proposition}
\begin{proof}
Let $x,y \in R_1 \cong (0,\infty)$, $\mu > 0$. Then
\begin{eqnarray*}
 && K(x,y,\mu^2) -  K_0(x,y,\mu^2) \\
 &=& \frac{1}{2i\mu}  \left(e^{i\mu|x-y|} -\left( 1+\frac{2(e^{i\mu L}+1)}{e^{i\mu L}-3} \right) \ e^{i\mu (x+y)}\right)
- \frac{1}{2i\mu}  \left(e^{i\mu |x-y|} + e^{i\mu (x+y)}\right)
 \\
   &=& \frac{1}{2i\mu} \left(  - \left(1+\frac{2(e^{i\mu L}+1)}{e^{i\mu L}-3}\right) -1 \right) e^{i\mu (x+y)}\\
   &=& - \, \frac{1}{\mu} \,  \frac{2(e^{i\mu L}-1)}{e^{i\mu L}-3} \, e^{i\mu (x+y)}.
\end{eqnarray*}
\end{proof}
By a simple substitution, we derive from E.~Stein \cite{Zygmung:99},
p.~334 the following variant of the Lemma of van der Corput for
$k=2$:
\begin{proposition} \label{corput variant}
Suppose that $\Phi:(a,b) \rightarrow \RR$ is smooth and satisfies
$\mid \Phi''(x) \mid \geq M > 0 $ for $x \in (a,b), \lambda > 0$ and that $\Psi
\in W^{1,1}(a,b).$ Then
\begin{equation}\label{corput}
  \bigg|\int_a^b e^{i \lambda \Phi(x)} \Psi (x) dx \bigg|
    \leq
    \frac{8}{(\lambda M)^{1/2}}\left(\mid \Psi(b) \mid + \int_a^b \mid \Psi' (x) \mid dx\right).
\end{equation}
\end{proposition}
Now we are able to compare the Schr\"{o}dinger time-evolution on the
queue of the tadpole and on the half line.

\begin{theorem} \label{difference kernel estimate}
Let $\left( e^{itH} \mathbb{I}_{(a,b)}(H)P_{ac}\right)(x,y)$ and
$\left( e^{itH_0} \mathbb{I}_{(a,b)}(H_0)\right)(x,y)$
be the kernels of the operator groups in the brackets. For
$0 \leq a < b < \infty$,   $t\ne 0$ and $x,y \in R_1 \cong (0,\infty)$
we have
\item  
  \begin{eqnarray}\label{rest}
\left( e^{itH} \mathbb{I}_{(a,b)}(H)P_{ac}\right)(x,y)
&-& \left( e^{itH_0} \mathbb{I}_{(a,b)}(H_0)\right)(x,y)  \\
&=& \int_{\sqrt{a}}^{\sqrt{b}} e^{i(t \mu^2 + \mu(x+y))}
\frac{4(1-e^{i\mu L})}{e^{i\mu L}-3} e^{i\mu (x+y)} d \mu,
\nonumber
  \end{eqnarray}
  and
 \begin{eqnarray}  \label{rest estimate}
\bigg| \left( e^{itH}\mathbb{I}_{(a,b)}(H)P_{ac}\right)(x,y)
&-& \left( e^{itH_0} \mathbb{I}_{(a,b)}(H_0)\right)(x,y)  \bigg|\\
& \leq & t^{-1/2} L  \ 2 \sqrt{2} \ \Bigl(4(2\sqrt{b}-\sqrt{a}) +
L(b-a)\Bigr).
\nonumber
  \end{eqnarray}
\end{theorem}
\begin{proof}
Ad \eqref{rest}: Thanks to proposition \ref{formula difference} we
have
 \begin{eqnarray*}
&&\left( e^{itH} \mathbb{I}_{(a,b)}(H)P_{ac} \right)(x,y)
- \left( e^{itH_0} \mathbb{I}_{(a,b)}(H_0)\right)(x,y)  \\
&=& \int_{a}^{b} e^{it \lambda }
\Bigl(K(x,y,\lambda ) -  K_0(x,y,\lambda)\Bigr) \ d \lambda  \\
&=& \int_{\sqrt{a}}^{\sqrt{b}} e^{it \mu^2}
\Bigl(K(x,y,\mu^2) - K_0(x,y,\mu^2)\Bigr) 2 \mu \ d \mu \\
&=& \int_{\sqrt{a}}^{\sqrt{b}} e^{i(t \mu^2 + \mu(x+y))}
\frac{4(1-e^{i\mu L})}{e^{i\mu L}-3}  d \mu.
  \end{eqnarray*}
Ad \eqref{rest estimate}: In view of applying proposition \ref{corput
variant}, we put
\begin{equation*}
    \Phi(\mu) = \mu^2 + \mu \frac{x+y}{t}  \quad \hbox{ and } \quad
    \psi(\mu) = \frac{4(1-e^{i\mu L})}{e^{i\mu L}-3}.
\end{equation*}
Then $\Phi''(\mu) = 2$  and thus
\begin{eqnarray}
\nonumber
 \bigg|\left( e^{itH} \mathbb{I}_{(a,b)}(H)P_{ac}\right)(x,y)
&-& \left( e^{itH_0} \mathbb{I}_{(a,b)}(H_0)\right)(x,y)  \bigg|\\
\nonumber & = &
 \bigg|\int_{\sqrt a}^{\sqrt b} e^{it\mu} \Psi (\mu) dx \bigg| \\
\label{appli corput}
 &\leq &
\frac{8}{(2t)^{1/2}}\left(\mid \Psi(\sqrt b) \mid
+ \int_{\sqrt a}^{\sqrt b}
\mid\Psi' (\mu) \mid d\mu \right).
  \end{eqnarray}
We estimate the expressions on the right hand side by using
\begin{equation*}
    \mid 1-e^{i\mu L} \mid \leq \mu L  \hbox{ and }
    \mid e^{i\mu L}-3 \mid \geq 2
\end{equation*}
which yields
\begin{equation*}
    \mid\Psi(\sqrt{b})\mid
    =
    \bigg| \frac{4(1-e^{i\sqrt{b} L})}{e^{i\sqrt{b} L}-3}\bigg|
    \leq
    \frac{4 \sqrt{b} L}{2} = 2 \sqrt{b} L.
\end{equation*}
Further we have
\begin{equation*}
    \Psi'(\mu) = \frac{-4Li e^{i\mu L}}{e^{i\mu L}-3} -
    \frac{-4Li (1 - e^{i\mu L})}{(e^{i\mu L}-3)^2}.
\end{equation*}
Therefore
\begin{equation*}
   \mid \Psi'(\mu) \mid \leq 2 L + \mu L^2 \hbox{ and thus }
   \int_{\sqrt{a}}^{\sqrt{b}} \mid \Psi'(\mu) \mid  d\mu \leq
   2L(\sqrt{b}-\sqrt{a}) + \frac{L^2}{2}(b-a).
\end{equation*}
Together with \eqref{appli corput} this yields the assertion.
\end{proof}
\begin{corollary} \label{dispersive perturbation estimatebis}
Let $0 \leq a < b < \infty$. Let $u_0 \in {\mathcal H} \cap
L^1(R_1)$ such that
\begin{equation}\label{support in queuebis}
    \hbox{\rm supp }u_0 \subset R_1 \ .
\end{equation}
Then  for all $t\ne 0$,  we have
\begin{eqnarray}
\nonumber \parallel e^{itH} \mathbb{I}_{(a,b)}(H)u_0
&-&  e^{itH_0} \mathbb{I}_{(a,b)}(H_0) u_0  \parallel_{L^\infty (R_1)}\\
\nonumber & \leq &
t^{-1/2} L  \ 2 \sqrt{2} \
\Bigl(4(2\sqrt{b}-\sqrt{a}) + L(b-a)\Bigr) \
\left\| u_0 \right\|_{L^1(R_1)}.
  \end{eqnarray}
\end{corollary}
\begin{proof}
  For $x\in R_1$, condition \eqref{support in queuebis} implies the second equality of
 \begin{eqnarray*}
&& (e^{itH} \mathbb{I}_{(a,b)}(H)\, u_0
- e^{itH_0} \mathbb{I}_{(a,b)}(H_0) \, u_0 )(x)\\
&=& \int_{{\mathcal R}}  \left( e^{itH}
\mathbb{I}_{(a,b)}(H)\right)(x,y) \, u_0(y) \ dy
- \int_{0}^\infty  \left( e^{itH_0}
\mathbb{I}_{(a,b)}(H_0)\right)(x,y)\, u_0(y) \ dy
 \\
&=& \int_{0}^\infty \Bigl[ \left( e^{itH}
\mathbb{I}_{(a,b)}(H)\right)(x,y) \, u_0(y)
-  \left( e^{itH_0}\mathbb{I}_{(a,b)}(H_0)\right)(x,y) \,
u_0(y)\Bigr] \ dy
\\
&=& \int_{0}^\infty \Bigl[ \left( e^{itH}
\mathbb{I}_{(a,b)}(H)P_{ac}\right)(x,y)
-  \left( e^{itH_0} \mathbb{I}_{(a,b)}(H_0)\right)(x,y) \Bigr]
u_0(y)\ dy.
  \end{eqnarray*}
The support condition \eqref{support in queuebis} implies $P_{pp}u_0 = 0$
and thus $u_0 = P_{ac} \, u_0$, which justifies the last equality.
Then the assertion follows from the hypothesis $u_0 \in {\mathcal H}
\cap L^1(R_1)$ and Theorem \ref{difference kernel estimate}
\eqref{rest estimate}.
\end{proof}

In Remark \ref{scale invariance} we proved that the tadpole problem
is scale invariant. In particular we showed, that if the dispersive
estimate in Theorem \ref{mainresult} holds with a constant $C$ for a
given circumference $L$ of the head of the tadpole, then it holds
for arbitrary $L$ with the same constant $C$.

Corollary \ref{dispersive perturbation estimatebis} of Theorem
\ref{difference kernel estimate} implies that the solution of the
Schr\"{o}dinger equation on the queue of the tadpole with an upper
frequency cutoff tends uniformly to the solution of the half-line
Neumann problem with the same upper frequency cutoff, if the support
of the initial condition has its support in the queue, after
compensation of the underlying $t^{-1/2}$-decay.

The upper frequency cutoff introduces in physical terms an upper
limit for the (group) velocity of wave packets and thus a lower
limit for the localization of wave packets (by an intuitive
application of the uncertainty principle). Thus wave packets are
large with respect to the head of the tadpole, if $L$ is small.  Therefore
the upper cutoff frequency destroys the scale invariance and it
becomes plausible, that the solutions of the tadpole problem tend to
the solutions of the half-line Neumann problem, if the head of the
tadpole shrinks to a point.

Technically this can be seen in inequality \eqref{appli corput} and
the subsequent arguments: we used
$$\mid 1-e^{i\mu L} \mid \leq \mu L $$
and the inequality of Stein (Proposition \ref{corput variant}),
which introduced the dependence of the cutoff frequency (by $\mu$)
and the perturbation aspect ($L \rightarrow 0$). By using the
triangular inequality which gives $\mid 1-e^{i\mu L} \mid \leq 2 $
and using the (pure) van der Corput estimate we would have avoided
the dependence on the upper cutoff frequency, but at the same time
also the perturbation aspect.

There exists also an interpretation of formula    \eqref{rest} from  Theorem
\ref{difference kernel estimate}  : by using the
expansion
\[
\frac{1}{e^{i\mu L}-3} =-\frac{1}{3}\sum_{k=0}^{+\infty}
\frac{e^{ik\mu  L}}{3^k},
\]
in the right hand side as in section \ref{proofm}, we obtain a
series representation of the difference of the solutions of the
tadpole problem on its queue and the half-line Neumann problem. The
terms correspond to signals passing from the head of the tadpole
into its queue after $k$ cycles around the head.


\end{document}